\documentclass{article}

\usepackage{arxiv}

\usepackage[utf8]{inputenc} 
\usepackage[T1]{fontenc}    
\usepackage{hyperref}       
\usepackage{url}            
\usepackage{booktabs}       
\usepackage{amsfonts}       
\usepackage{nicefrac}       
\usepackage{microtype}      
\usepackage{lipsum}
\usepackage{graphicx}
\graphicspath{ {./images/} }
\usepackage{amssymb}
\usepackage{amsmath}

\usepackage{subcaption}
\usepackage{xcolor}
\usepackage{tabularx}
\usepackage[ruled,linesnumbered]{algorithm2e}
\usepackage{amsthm}
\newtheorem{lemma}{Lemma}[section]       
\theoremstyle{definition}                
\usepackage{nomencl}
\makenomenclature
\usepackage{etoolbox}
\renewcommand\nomgroup[1]{%
	\item[\bfseries
	\ifstrequal{#1}{R}{Roman Symbols}{%
		\ifstrequal{#1}{G}{Greek Symbols}{%
			\ifstrequal{#1}{A}{Abbreviations}{%
				\ifstrequal{#1}{S}{Subscripts}}}}%
	]}

\title{Analytical Extraction of Conditional Aleatory Sensitivities Across Epistemic Space via a Single PCE Model}

\author{
	Shijie Zhong \\
	School of Power and Energy\\
	Northwestern Polytechnical University\\
	Xi'an, Shanxi 710129 \\
	\texttt{zhongsj@mail.nwpu.edu.cn} \\
	\And
	Huiyou Tan \\
    School of Power and Energy\\
    Northwestern Polytechnical University\\
    Xi'an, Shanxi 710129 \\
    \texttt{tanhy@mail.nwpu.edu.cn} \\
	\And
	Jiangfeng Fu \\
	School of Power and Energy\\
	Northwestern Polytechnical University\\
	Xi'an, Shanxi 710129 \\
	\texttt{fjf@nwpu.edu.cn} \\
}

\begin{document}
\maketitle
\begin{abstract}
	In hybrid uncertainty quantification, evaluating how aleatory sensitivities vary under epistemic uncertainty, referred to as conditional Sobol' indices, is typically hindered by the computationally expensive double-loop procedure. Classical Polynomial Chaos Expansion (PCE) provides efficient access to global sensitivity measures but cannot directly resolve sensitivity variation across the epistemic space without repeated surrogate reconstruction. This study proposes a unified Bayesian framework that extracts continuous conditional Sobol' fields from a single global PCE representation. By exploiting the tensor-product structure of orthogonal polynomial bases in an augmented stochastic space, the global expansion is analytically decomposed into epistemic-dependent coefficient fields, enabling a closed-form variance decomposition. As a result, conditional Sobol' indices can be computed through a purely algebraic post-processing step without additional model evaluations or retraining. In addition, a Reversible Jump Markov Chain Monte Carlo (RJMCMC) scheme is incorporated to perform adaptive basis selection and trans-dimensional inference, while simultaneously providing Bayesian credible intervals for the conditional sensitivity measures. Numerical experiments on a high-dimensional groundwater flow model demonstrate that the proposed method significantly reduces computational cost while maintaining smooth sensitivity fields and statistically consistent uncertainty quantification across the epistemic domain.
\end{abstract}

\section{Introduction}

Hybrid uncertainty arises from the presence of both aleatory and epistemic uncertainties, requiring rigorous uncertainty quantification frameworks \cite{der2009aleatory, hoffman1994propagation}. To avoid misleading reliability estimates caused by mixing these two sources, non-probabilistic frameworks such as interval analysis, evidence theory, and probability boxes (p-boxes) have been developed to enable robust assessments under incomplete information \cite{klir2003uncertainty, ferson2003constructing}. However, propagating hybrid uncertainties in complex engineering systems is often computationally expensive due to double-loop procedures. To reduce this computational cost, surrogate models and multi-fidelity techniques have been widely used for reliability analysis and design optimization \cite{liu2024hybrid}.

Given such hybrid representations, Global Sensitivity Analysis (GSA), particularly variance-based Sobol' indices, is widely used to quantify the contribution of input uncertainties to output variability \cite{sobol1993sensitivity, angus1994probability}. Classical GSA, however, is formulated for scalar outputs under fully probabilistic assumptions, whereas hybrid uncertainty problems often involve outputs that depend on epistemic parameters. In such cases, the sensitivity of aleatory variables becomes conditional on the epistemic state and may vary throughout the epistemic space. Although recent surrogate-based and augmented-space methods have significantly improved the efficiency of hybrid uncertainty analysis and can provide global or aggregated Sobol' indices under hybrid uncertainties \cite{liu2025efficient, luo2025unified}, they do not explicitly represent conditional variance decompositions. As a result, continuous conditional Sobol' fields are generally not available in analytical form, and a framework that directly connects hybrid uncertainty representations with conditional sensitivity measures is still lacking.

Built on decoupled or fused parameter representations, Polynomial Chaos Expansion (PCE) has become a widely adopted method for hybrid uncertainty quantification due to its solid mathematical foundation \cite{ghanem2003stochastic}. Following early efforts that incorporated interval-based epistemic uncertainties into unified PCE formulations \cite{monti2009extending}, recent studies have reduced the curse of dimensionality through dimension fusion and polymorphic variables \cite{yusuf2021polymorphic, prasad2017novel}. By representing mixed uncertainties with a single variable, these approaches successfully yield highly compact surrogate models.

However, this mathematical compression introduces several critical limitations. First, parameter fusion inherently obscures the distinction between epistemic and aleatory uncertainties, hindering the interpretation of their individual contributions. Second, the resulting mixed variables often follow nonstandard probability distributions that are incompatible with the classical orthogonal polynomial framework (the Askey scheme \cite{xiu2002wiener}), requiring the cumbersome numerical construction of custom orthogonal bases.

Consequently, maintaining a decoupled parameter representation remains highly desirable to preserve both physical interpretability and mathematical elegance. A mathematically rigorous approach to achieve this is through the Probability Integral Transform (PIT), which introduces an independent auxiliary variable to isolate pure aleatory variability from epistemic parameter uncertainty \cite{sankararaman2013separating}. Yet, even within such PIT-based decoupled frameworks, existing PCE formulations are not designed for the analytical construction of continuous conditional Sobol' fields. Since conditional sensitivity measures are not explicitly represented in the resulting expansions, they must typically be recovered through computationally expensive repeated evaluations or brute-force sampling across the epistemic space.

Furthermore, reverting to a decoupled representation inevitably reintroduces the curse of dimensionality. Although deterministic sparse algorithms such as Orthogonal Matching Pursuit \cite{tropp2007signal} and Least Angle Regression \cite{Blatman2011AdaptiveSP} improve the scalability of high-dimensional PCE, they provide only point estimates of the expansion coefficients due to their deterministic basis selection and least-squares formulation. Bayesian sparse PCE \cite{shao2017bayesian} offers an alternative by casting coefficient estimation as a probabilistic inference problem. More recently, Reversible Jump Markov Chain Monte Carlo (RJMCMC)-based formulations have further advanced this framework \cite{rumsey2026bayesian}, in which both the model structure and coefficients are treated as unknowns. This enables simultaneous basis selection and parameter inference, yielding full posterior distributions and rigorous credible intervals for both the surrogate model and the resulting conditional Sobol' indices.

To address these limitations, this paper proposes a unified Bayesian framework for the analytical construction of continuous conditional Sobol' fields under hybrid uncertainty. An inverse probability integral transform is first used to map the hybrid uncertainty space into a decoupled latent space, preserving the tensor-product structure of the polynomial basis. Within this representation, conditional Sobol' indices are shown to be directly recoverable from the coefficients of a global PCE surrogate through purely algebraic operations. This enables the construction of continuous conditional sensitivity fields without additional model evaluations or surrogate reconstruction. The framework is further combined with an RJMCMC-based sparse PCE formulation to quantify surrogate-induced uncertainty and provide posterior estimates with Bayesian credible intervals. Numerical examples, including a groundwater flow model, demonstrate the accuracy, robustness, and efficiency of the proposed approach.

The remainder of this paper is organized as follows. Section \ref{sec2} introduces the theoretical background, including hybrid uncertainty modeling via the probability integral transform, the fundamentals of Polynomial Chaos Expansion (PCE), and the trans-dimensional Bayesian inference framework. Section \ref{sec3} presents the core methodology, detailing the mathematical proofs for the analytical extraction of conditional Sobol' indices via basis decomposition, along with the complete algorithmic implementation using RJMCMC-PCE. In Section \ref{sec4}, the accuracy, robustness, and efficiency of the proposed framework are demonstrated through a rigorous analytical benchmark and a high-dimensional borehole groundwater flow model.

\section{Backgroud}
\label{sec2}
\subsection{Hybrid Uncertainty Modeling Based on the Probability Integral Transform}

In practical engineering systems, uncertainties are described within a fixed model universe consisting of basic random variables, probabilistic sub-models, and physical response models. Within this framework, it is convenient to distinguish between aleatory and epistemic uncertainty, where the distinction depends on how the model is constructed rather than on intrinsic properties of the physical world.

Let \( \theta \in \Theta \subset \mathbb{R}^d \) denote uncertain parameters associated with the probabilistic description of the system, reflecting incomplete knowledge in the modeling process. Conditional on \( \theta \), the system input is represented by a probabilistic model of the form
\begin{equation}
	X \sim p(X \mid \theta),
\end{equation}
where the distribution is interpreted as part of the chosen model structure within the model universe.

Within the same framework, the system response is defined through a deterministic physical model
\begin{equation}
	Y = g(X),
\end{equation}
where all uncertainty is embedded in the probabilistic characterization of \(X\) and the uncertainty in \(\theta\). In this representation, aleatory uncertainty refers to the variability of \(X\) under a fixed model specification, while epistemic uncertainty corresponds to the uncertainty in the model parameters \(\theta\), which may in principle be reduced through additional information or model refinement \cite{der2009aleatory}.

In the context of epistemic uncertainty, a key idea is that when information is incomplete or data are limited, probabilities do not need to be represented as single precise values. Instead, they can be modeled as intervals or sets, where uncertainty is expressed through lower and upper bounds or a family of admissible probability distributions, reflecting incomplete knowledge about the true probability measure \cite{walley1991statistical}.

At this stage, a useful interpretation emerges: the separation between aleatory variability and epistemic uncertainty can be further formalized through the introduction of an auxiliary variable, which enables a fully deterministic reformulation of the stochastic system\cite{sankararaman2013separating}. In this context, the auxiliary variable is introduced to explicitly represent the intrinsic variability of the input random variable while decoupling it from distributional parameter uncertainty. For each uncertain input variable $X$, the probability integral transform establishes a one-to-one mapping between realizations of $X$ and a standard uniform variable $\xi \sim \mathcal{U}(0,1)$, such that
\begin{equation}
	\xi = F_X(X \mid \theta).
\end{equation}

Consequently, the input can be equivalently expressed as
\begin{equation}
	X = F_X^{-1}(\xi \mid \theta),
\end{equation}
where $\xi$ serves as an auxiliary variable that is statistically independent of the epistemic parameters $\theta$. This transformation provides two important advantages: (i) it yields an explicit representation of the aleatory variability via a distribution-free latent variable $\xi$, and (ii) it ensures a deterministic mapping from the joint space $(\xi, \theta)$ to the system response.

In this formulation, the overall uncertainty in $X$ is naturally decomposed into two components: the intrinsic variability captured by $\xi$, and the epistemic uncertainty associated with the parameter vector $\theta$. This separation aligns with the classical double-loop interpretation in uncertainty propagation, while simultaneously enabling a unified deterministic representation of the model response. As a result, the system output can be written as
\begin{equation}
	Y = g\!\left(F_X^{-1}(\xi \mid \theta)\right) = h(\xi,\theta),
\end{equation}
which forms the basis for global sensitivity analysis and uncertainty decomposition.

\subsection{Theoretical foundation of PCE}
Consider an $M$-dimensional random vector $\boldsymbol{X}$, whose components are mutually independent and have a joint probability density function $f_{\boldsymbol{X}}$. Let the computational model $\mathcal{M}(\cdot)$ map the random vector of the input to the random variable of the output $Y = \mathcal{M}(\boldsymbol{X})$. Assuming $Y$ has a finite variance, based on PCE, $Y$ can be represented as:

\begin{equation}
	\label{PCE}
	Y=\mathcal{M} \left( \boldsymbol{X} \right) =\sum_{\boldsymbol{\alpha }\in \mathbb{N} ^M}{y_{\boldsymbol{\alpha }}\Psi _{\boldsymbol{\alpha }}\left( \boldsymbol{X} \right)}
\end{equation}
where $\boldsymbol{\alpha}$ is the multi-index, $y_{\boldsymbol{\alpha}}$ are the corresponding expansion coefficients, and $\Psi_{\boldsymbol{\alpha}}(\boldsymbol{X})$ are orthonormal multivariate polynomials with respect to $f_{\boldsymbol X}$. These polynomials are typically constructed by the tensor product of one-dimensional orthonormal polynomials:
\begin{equation}
	\Psi_{\boldsymbol{\alpha}}(\boldsymbol{X}) = \prod_{i=1}^M \phi_{\alpha_i}(x_i)
\end{equation}
which satisfy the orthonormality condition $\left\langle \Psi_{\boldsymbol{\alpha}}, \Psi_{\boldsymbol{\beta}} \right\rangle = \delta_{\boldsymbol{\alpha} \boldsymbol{\beta}}$, where $\delta_{\boldsymbol{\alpha} \boldsymbol{\beta}}$ is the multi-dimensional Kronecker symbol.

For uniformly distributed input variables defined on $[-1, 1]$, the one-dimensional polynomial basis $\phi_k(x)$ is constructed based on the corresponding probability density function using the standard orthogonal Legendre polynomials \cite{xiu2002wiener}. These basis functions are defined as:

\begin{equation}
	\phi_k(x) = \frac{P_k(x)}{\sqrt{\gamma_k}}
\end{equation}

where $P_k(x)$ is the $k$-th order Legendre polynomial, and $\gamma_k = 1 / (2k+1)$ is the corresponding normalization factor. This ensures the orthonormality condition $\left\langle \phi_i(x), \phi_j(x) \right\rangle = \delta_{ij}$.

Furthermore, the standard orthogonal Legendre polynomials can be calculated using the following recurrence formula:

\begin{equation}
	(k+1)P_{k+1}(x) = (2k+1)xP_k(x) - kP_{k-1}(x)
\end{equation}

In practical applications, for computational convenience, the PCE is usually limited such that the total polynomial order does not exceed $p$:
\begin{equation}
	\mathcal{A} ^{M,p}=\left\{ \boldsymbol{\alpha }\in \mathbb{N} ^M:\left| \boldsymbol{\alpha } \right|\leqslant p \right\}
\end{equation}

The number of terms, $P$, in the truncated series is:
\begin{equation}
	\label{Card}
	P=\mathrm{card}\mathcal{A} ^{M,p}=\left( \begin{array}{c}
		M+p\\
		p\\
	\end{array} \right) 
\end{equation}

The truncated PCE is then expressed as:
\begin{equation}
	\label{tunPCE}
	\mathcal{M} ^{PC}\left( \boldsymbol{X} \right) =\sum_{\boldsymbol{\alpha }\in \mathcal{A} ^{M,p}}{y_{\boldsymbol{\alpha }}\Psi _{\boldsymbol{\alpha }}\left( \boldsymbol{X} \right)}
\end{equation}

Consider the truncated PCE representation of a model output Eq.~\ref{tunPCE}, the unique Sobol' decomposition characterizes the model response as a sum of functions of increasing dimensions \cite{sudret2008global, xiu2010numerical}. 
	
The PCE terms can be reorganized according to the subset of input variables they involve. For any non-empty subset $\mathbf{u} \subset \{1, \dots, M\}$, define the associated index set:
\begin{equation}
	\mathcal{A}_{\mathbf{u}} = \left\{ \boldsymbol{\alpha} \in \mathcal{A} : \alpha_k \neq 0 \iff k \in \mathbf{u} \right\}
\end{equation}
	
Using this partition, the truncated PCE can be rewritten as a sum of contributions associated with different variable subsets:
\begin{equation}
	\mathcal{M} ^{PC} = y_{\mathbf{0}} + \sum_{\substack{\mathbf{u} \subset \{1,\dots,M\} \\ \mathbf{u} \neq \emptyset}} 
	\sum_{\boldsymbol{\alpha} \in \mathcal{A}_{\mathbf{u}}} 
		y_{\boldsymbol{\alpha}} \, \Psi_{\boldsymbol{\alpha}}(\boldsymbol{\xi})
\end{equation}
	
Due to the orthonormality of the polynomial basis, the total variance of the PCE model and the partial variance associated with a subset of variables $\mathbf{u}$ can be computed directly from the expansion coefficients as:
\begin{equation}
	\mathrm{Var}[\mathcal{M} ^{PC}] = \sum_{\substack{\boldsymbol{\alpha} \in \mathcal{A} \\ \boldsymbol{\alpha} \neq \mathbf{0}}} y_{\boldsymbol{\alpha}}^2
\end{equation}
\begin{equation}
	\mathrm{Var}[\mathcal{M} ^{PC}_{\mathbf{u}}] = \sum_{\boldsymbol{\alpha} \in \mathcal{A}_{\mathbf{u}}} y_{\boldsymbol{\alpha}}^2
\end{equation}
	
The Sobol' index corresponding to $\mathbf{u}$ is therefore obtained as:
\begin{equation}
	S_{\mathbf{u}} = \frac{V_{\mathbf{u}}}{V} 
	= \frac{\sum_{\boldsymbol{\alpha} \in \mathcal{A}_{\mathbf{u}}} y_{\boldsymbol{\alpha}}^2}
	{\sum_{\substack{\boldsymbol{\alpha} \in \mathcal{A} \\ \boldsymbol{\alpha} \neq \mathbf{0}}} y_{\boldsymbol{\alpha}}^2}
\end{equation}
which enables the direct computation of Sobol' indices from PCE coefficients without additional model evaluations.

While these global indices $S_{\mathbf{u}}$ provide a comprehensive overview of variable importance, they are static and cannot reflect how sensitivities vary under specific conditions or within localized sub-domains. This limitation motivates the development of the analytical conditional framework presented in the next section.

\subsection{Bayesian inference and trans-dimensional formulation}
To rigorously quantify the epistemic uncertainty and achieve adaptive sparsity, the PCE framework is cast into a fully Bayesian probabilistic model. The deterministic model evaluations are assumed to be observed with a Gaussian error term, formulating the response as:
\begin{equation}
	y_i = \mathcal{M}^{PC}(\boldsymbol{x}_i) + \epsilon_i, \quad \epsilon_i \sim \mathcal{N}(0, \sigma^2)
\end{equation}
where $\mathcal{M}^{PC}(\boldsymbol{x}) = \beta_0 + \sum_{m=1}^M \beta_m \Psi_m(\boldsymbol{x}|\boldsymbol{\alpha}_m)$ represents the truncated PCE with $M$ active basis functions, and $\sigma^2$ captures the residual variance.

The adaptive nature of the RJPCE algorithm is driven by treating the number of basis functions $M$ and their corresponding multi-indices $\mathcal{A} = \{\boldsymbol{\alpha}_1, \dots, \boldsymbol{\alpha}_M\}$ as unknown random variables. To control the model's complexity dynamically, a hierarchical Poisson prior is placed on the model dimension $M$:
\begin{equation}
	M | \lambda \sim \text{Poiss}(\lambda), \quad \lambda \sim \text{Gamma}(a_M, b_M)
\end{equation}
which naturally penalizes overly complex representations and encourages sparsity \cite{rumsey2026bayesian, francom2020bass}.

For the expansion coefficients $\boldsymbol{\beta} = (\beta_0, \beta_1, \dots, \beta_M)^\top$, a modified $g$-prior is assigned to induce complexity-dependent shrinkage. This specific prior is defined as:
\begin{equation}
	\boldsymbol{\beta} | M, \sigma^2, g_0^2 \sim \mathcal{N}_{M+1}\left(\mathbf{0}, \sigma^2 g_0^2 \boldsymbol{D}(g)(\boldsymbol{\Psi}^\top \boldsymbol{\Psi})^{-1} \boldsymbol{D}(g)\right)
\end{equation}
where $g_0^2 \sim \text{Inv-Gamma}(a_g, b_g)$ acts as a global regularizer, and $\boldsymbol{D}(g)$ is a diagonal matrix containing penalty weights $g_m$ that apply stronger regularization to basis functions with higher interaction orders and polynomial degrees \cite{zhang2015two, zellner1986assessing}. The residual variance is assigned a standard conjugate prior $\sigma^2 \sim \text{Inv-Gamma}(a_\sigma, b_\sigma)$ \cite{hoff2009first}.

Since the model dimension $M$ varies during the learning process, standard Markov Chain Monte Carlo methods are insufficient. A Reversible Jump MCMC (RJMCMC) algorithm is utilized to explore the trans-dimensional parameter space through structural proposals, including birth, death, and mutations \cite{green2001delayed}. Suppose the chain proposes a move from the current model $\mathcal{M}_{cur}$ to a candidate model $\mathcal{M}_{cand}$ via move type $X$. The proposal is accepted with probability $\min(1, \alpha_X)$, where the log-acceptance ratio is given by:
\begin{equation}
	\log \alpha_X = \log\left(\frac{p(\boldsymbol{y}|\mathcal{M}_{cand})}{p(\boldsymbol{y}|\mathcal{M}_{cur})}\right) + \log\left(\frac{p(\mathcal{M}_{cand})}{p(\mathcal{M}_{cur})}\right) + \log A_X
\end{equation}
These three terms represent the log-marginal likelihood ratio, the log-prior ratio of the model structures, and the asymmetric proposal probability ratio $A_X$ associated with the specific move type, respectively.

Conditioned on the accepted model structure at each RJMCMC iteration, the continuous parameters are updated using exact conjugate Gibbs sampling. For instance, the full conditional posterior for the regression coefficients follows a multivariate normal distribution:
\begin{equation}
	\boldsymbol{\beta} | \cdot \sim \mathcal{N}(\boldsymbol{\mu}_n, \sigma^2 \boldsymbol{\Sigma}_n)
\end{equation}
where the posterior covariance is $\boldsymbol{\Sigma}_n = (\boldsymbol{\Psi}^\top \boldsymbol{\Psi} + \boldsymbol{S}_0^{-1})^{-1}$ and the posterior mean is $\boldsymbol{\mu}_n = \boldsymbol{\Sigma}_n \boldsymbol{\Psi}^\top \boldsymbol{y}$, with $\boldsymbol{S}_0$ derived directly from the modified $g$-prior. This analytical tractability within fixed dimensions ensures efficient chain mixing and robust posterior sampling.

\section{Analytical conditional Sobol' analysis via basis decomposition}
\label{sec3}
In hybrid uncertainty quantification, the model response is parameterized by a subset of epistemic variables $\boldsymbol{\theta}$ across a continuous parameter space. Rather than constructing independent surrogates for each discrete epistemic realization, a more efficient and rigorous approach is to treat both the epistemic variables $\boldsymbol{\theta}$ and the aleatory random vector $\boldsymbol{\xi}$ (or their mapped equivalents in the latent space, $\boldsymbol{u}_\Theta$ and $\boldsymbol{u}_\xi$) as components of a joint augmented input vector $\boldsymbol{X}=\left( \boldsymbol{\theta}, \boldsymbol{\xi} \right)^\mathrm{T}$. By doing so, a unified PCE model can be constructed, providing a continuous and integrated functional representation of the system response over the entire hybrid uncertainty domain.

\subsection{Analytical extraction of conditional Sobol' indices}

To extract the conditional sensitivity structure with respect to any specific epistemic condition $\boldsymbol{\theta}$, the multi-index $\boldsymbol{\alpha}$ is partitioned into two disjoint subsets: $\boldsymbol{\alpha}_K$, which clusters indices related to the epistemic parameters $\boldsymbol{\theta}$; and $\boldsymbol{\alpha}_L$, which clusters the indices associated with the aleatory variables $\boldsymbol{\xi}$. Leveraging the tensor-product nature of the PCE basis functions, the joint basis $\Psi_{\boldsymbol{\alpha}}(\boldsymbol{X})$ can be analytically factorized into a conditional form:
\begin{equation}
	\Psi _{\boldsymbol{\alpha }}\left( \boldsymbol{X} \right) =\Psi _{\boldsymbol{\alpha }}\left( \boldsymbol{\theta},\boldsymbol{\xi } \right) =\Psi _{\boldsymbol{\alpha }_K}\left( \boldsymbol{\theta} \right) \Psi _{\boldsymbol{\alpha }_L}\left( \boldsymbol{\xi } \right)
\end{equation}
By regrouping terms based on the stochastic multi-indices $\boldsymbol{\alpha}_L$, the global PCE model is reformulated into an analytical conditional representation. This transformation reveals that the influence of the epistemic conditions can be entirely encapsulated within a new set of varying coefficients. The formal mathematical properties of this decomposition are summarized in the following lemma.

\begin{lemma}[Orthogonality preservation for analytical conditional PCE]
	\label{lem:ortho}
	Given the basis factorization $\Psi_{\boldsymbol{\alpha}}(\boldsymbol{\theta}, \boldsymbol{\xi}) = \Psi_{\boldsymbol{\alpha}_K}(\boldsymbol{\theta}) \Psi_{\boldsymbol{\alpha}_L}(\boldsymbol{\xi})$, the global PCE model can be analytically reformulated into a conditional expansion:
	\begin{equation}
		\mathcal{M}^{PC}(\boldsymbol{\xi} \mid \boldsymbol{\theta}) = \sum_{\boldsymbol{\alpha}_L} c_{\boldsymbol{\alpha}_L}(\boldsymbol{\theta}) \Psi_{\boldsymbol{\alpha}_L}(\boldsymbol{\xi})
	\end{equation}
	where the parametric coefficient fields $c_{\boldsymbol{\alpha}_L}(\boldsymbol{\theta})$ are defined as:
	\begin{equation}
		c_{\boldsymbol{\alpha}_L}(\boldsymbol{\theta}) = \sum_{\boldsymbol{\alpha}_K} y_{\boldsymbol{\alpha}_K, \boldsymbol{\alpha}_L} \Psi_{\boldsymbol{\alpha}_K}(\boldsymbol{\theta})
	\end{equation}
	Crucially, for any fixed value of the epistemic variables $\boldsymbol{\theta}$, the basis functions $\Psi_{\boldsymbol{\alpha}_L}(\boldsymbol{\xi})$ preserve their orthonormality with respect to the aleatory random vector $\boldsymbol{\xi}$.
\end{lemma}

\begin{proof}
	Starting from the PCE formulation, we substitute the factorized basis functions $\Psi_{\boldsymbol{\alpha}}(\boldsymbol{\theta}, \boldsymbol{\xi}) = \Psi_{\boldsymbol{\alpha}_K}(\boldsymbol{\theta}) \Psi_{\boldsymbol{\alpha}_L}(\boldsymbol{\xi})$ into the expansion:
	\begin{equation}
		\mathcal{M}^{PC}(\boldsymbol{\xi} \mid \boldsymbol{\theta}) = \sum_{\boldsymbol{\alpha}_L} \sum_{\boldsymbol{\alpha}_K} y_{\boldsymbol{\alpha}_K, \boldsymbol{\alpha}_L} \Psi_{\boldsymbol{\alpha}_K}(\boldsymbol{\theta}) \Psi_{\boldsymbol{\alpha}_L}(\boldsymbol{\xi})
	\end{equation}
	By invoking the linearity of the summation, the terms associated with the epistemic basis functions $\Psi_{\boldsymbol{\alpha}_K}(\boldsymbol{\theta})$ can be grouped:
	\begin{equation}
		\mathcal{M}^{PC}(\boldsymbol{\xi} \mid \boldsymbol{\theta}) = \sum_{\boldsymbol{\alpha}_L} \left[ \sum_{\boldsymbol{\alpha}_K} y_{\boldsymbol{\alpha}_K, \boldsymbol{\alpha}_L} \Psi_{\boldsymbol{\alpha}_K}(\boldsymbol{\theta}) \right] \Psi_{\boldsymbol{\alpha}_L}(\boldsymbol{\xi})
	\end{equation}
	By defining the term in the brackets as the parametric coefficient $c_{\boldsymbol{\alpha}_L}(\boldsymbol{\theta})$, we obtain the conditional form of the expansion.
	
	Regarding the orthogonality, consider the inner product with respect to the random vector $\boldsymbol{\xi}$ for a given $\boldsymbol{\theta}$. Since the functions $\Psi_{\boldsymbol{\alpha}_L}(\boldsymbol{\xi})$ are elements of the original orthonormal basis set $\mathcal{A}$ restricted to the stochastic dimensions of $\boldsymbol{\xi}$, their orthonormality is inherently preserved under the conditional probability measure:
	\begin{equation}
		\langle \Psi_{\boldsymbol{\alpha}_L}, \Psi_{\boldsymbol{\beta}_L} \rangle_{\boldsymbol{\xi}} = \delta_{\boldsymbol{\alpha}_L \boldsymbol{\beta}_L}
	\end{equation}
\end{proof}

A key property of this analytical conditional representation (Lemma~\ref{lem:ortho}) is that the orthonormality of the polynomial basis is preserved with respect to the aleatory variables $\boldsymbol{\xi}$. This enables the direct computation of statistical moments conditioned on any given epistemic state $\boldsymbol{\theta}$. Specifically, the conditional mean of the model response, which represents the deterministic trend driven by epistemic parameters, is directly given by the zero-order coefficient field:
\begin{equation}
	\mathbb{E} \left[ \mathcal{M}^{PC}(\boldsymbol{\xi} \mid \boldsymbol{\theta}) \right] = c_{\boldsymbol{0}}(\boldsymbol{\theta}) = \sum_{\boldsymbol{\alpha}_K} y_{\boldsymbol{\alpha}_K, \boldsymbol{0}} \Psi_{\boldsymbol{\alpha}_K}(\boldsymbol{\theta})
\end{equation}
Subsequently, the conditional variance, which characterizes the aleatory uncertainty magnitude at a given $\boldsymbol{\theta}$, is obtained by summing the squares of all higher-order coefficient fields:
\begin{equation}
	\mathrm{Var} \left[ \mathcal{M}^{PC}(\boldsymbol{\xi} \mid \boldsymbol{\theta}) \right]= \sum_{\boldsymbol{\alpha}_L \ne \boldsymbol{0}} c_{\boldsymbol{\alpha}_L}^{2}(\boldsymbol{\theta}),\label{SCvar}
\end{equation}
These expressions demonstrate that both the mean and variance are no longer static scalars, but are expressed as continuous, analytical functions of the epistemic variables $\boldsymbol{\theta}$.

In the framework of the analytical conditional PCE, Sobol' sensitivity indices can be derived directly from the parametric coefficient fields $c_{\boldsymbol{\alpha}_L}(\boldsymbol{\theta})$, completely bypassing the need for additional model evaluations or sampling. To obtain the conditional sensitivity distribution, we focus on the sensitivity of the response with respect to the aleatory input vector $\boldsymbol{\xi} = \{\xi_1, \dots, \xi_L\}$ conditional on $\boldsymbol{\theta}$. For any non-empty subset of indices $\mathbf{u} \subset \{1, \dots, L\}$, we define the associated conditional multi-index set:
\begin{equation}
	\mathcal{A}_{\mathbf{u}, L}^{M, p} = \{ \boldsymbol{\alpha}_L \in \mathbb{N}^L : \alpha_k \neq 0 \Longleftrightarrow k \in \mathbf{u}, \quad \text{and } (\boldsymbol{\alpha}_K, \boldsymbol{\alpha}_L) \in \mathcal{A}^{M, p}\}
\end{equation}
Based on these definitions, the analytical conditional Sobol' decomposition is formalized in the following lemma.

\begin{lemma}[Analytical conditional Sobol' decomposition]\label{lem:scsobol}
	Given the conditional PCE model $\mathcal{M}^{PC}(\boldsymbol{\xi} \mid \boldsymbol{\theta})$, for any fixed epistemic realization $\boldsymbol{\theta}$, there exists a unique Sobol' decomposition with respect to the aleatory input vector $\boldsymbol{\xi}$:
	\begin{equation}
		\mathcal{M} ^{PC}\left( \boldsymbol{\xi } \mid \boldsymbol{\theta} \right) = \mathcal{M} ^{PC}_{\boldsymbol{0}}(\boldsymbol{\theta})+\sum_{ \mathbf{u}\ne \emptyset}{\mathcal{M} ^{PC}_{\mathbf{u}}\left( \boldsymbol{\xi}_{\mathbf{u}} \mid \boldsymbol{\theta} \right)}
	\end{equation}
	where $\mathcal{M}^{PC}_{\boldsymbol{0}}(\boldsymbol{\theta})$ is the conditional mean, and each component $\mathcal{M}_{\mathbf{u}}$ is represented by the corresponding $\text{PCE}$ terms:
	\begin{equation}
		\mathcal{M} ^{PC} _{\mathbf{u}}\left( \boldsymbol{\xi }_{\mathbf{u}} \mid \boldsymbol{\theta} \right) = \sum_{\boldsymbol{\alpha }_L\in \mathcal{A} _{\mathbf{u},L}^{M,p}}{c_{\boldsymbol{\alpha }_L}(\boldsymbol{\theta})\Psi _{\boldsymbol{\alpha }_L}\left( \boldsymbol{\xi } \right)}
	\end{equation}
	
	Under the orthonormality of the polynomial chaos basis $\Psi_{\boldsymbol{\alpha}_L}$, the conditional Sobol' indices are obtained analytically from the coefficient fields as:
	
	\begin{enumerate}
		\item {First-order sensitivity index}
		\begin{equation}
			S_{i \mid \boldsymbol{\theta}} 
			= \frac{\sum_{\boldsymbol{\alpha}_L \in \mathcal{A}_{\mathbf{u},L}^{M,p}, \ \mathbf{u} = \{i\}} 
				c_{\boldsymbol{\alpha}_L}^2(\boldsymbol{\theta})}
			{\sum_{\substack{\boldsymbol{\alpha}_L \in \mathcal{A}^{M,p} \\ \boldsymbol{\alpha}_L \neq \boldsymbol{0}}} 
				c_{\boldsymbol{\alpha}_L}^2(\boldsymbol{\theta})}
		\end{equation}
		
		\item {Interaction sensitivity index}
		\begin{equation}
			S_{\mathbf{u} \mid \boldsymbol{\theta}} 
			= \frac{\sum_{\boldsymbol{\alpha}_L \in \mathcal{A}_{\mathbf{u},L}^{M,p}} 
				c_{\boldsymbol{\alpha}_L}^2(\boldsymbol{\theta})}
			{\sum_{\substack{\boldsymbol{\alpha}_L \in \mathcal{A}^{M,p} \\ \boldsymbol{\alpha}_L \neq \boldsymbol{0}}} 
				c_{\boldsymbol{\alpha}_L}^2(\boldsymbol{\theta})}
		\end{equation}
		
		\item {Total sensitivity index}
		\begin{equation}
			S_{i \mid \boldsymbol{\theta}}^{T}
			= \frac{\sum_{\substack{\boldsymbol{\alpha}_L \in \mathcal{A}^{M,p} \\ \alpha_i \neq 0}} 
				c_{\boldsymbol{\alpha}_L}^2(\boldsymbol{\theta})}
			{\sum_{\substack{\boldsymbol{\alpha}_L \in \mathcal{A}^{M,p} \\ \boldsymbol{\alpha}_L \neq \boldsymbol{0}}} 
				c_{\boldsymbol{\alpha}_L}^2(\boldsymbol{\theta})}
			\label{sensi}
		\end{equation}
	\end{enumerate}
\end{lemma}

\begin{proof}
	The result follows from the uniqueness of the Sobol' decomposition and the orthogonality of the basis functions $\Psi_{\boldsymbol{\alpha}_L}$ established in Lemma~\ref{lem:ortho}. By partitioning the conditional terms into the index sets $\mathcal{A}_{\mathbf{u}, L}$, the model response is decomposed into mutually orthogonal components. The conditional variance contributions are then obtained as sums of the squared parametric coefficients $c_{\boldsymbol{\alpha}_L}(\boldsymbol{\theta})$. Normalization by the total conditional variance (Eq.~\ref{SCvar}) yields the analytical sensitivity indices.
\end{proof}

\subsection{Isoprobabilistic Transformation and RJMCMC-PCE Implementation}
\label{subsec:algorithmic_implementation}

To circumvent the intractable computational bottleneck imposed by traditional double-loop nested Monte Carlo simulations or repetitive surrogate modeling, we propose a fully vectorized, single-loop algorithmic architecture. This framework evaluates the computationally expensive physics model strictly in the physical parameter space, while conducting the Bayesian adaptive Polynomial Chaos Expansion (PCE) and subsequent variance decomposition entirely within a decoupled, orthogonal latent space. The implementation consists of five sequential stages:

\textbf{1. Joint Sampling in the Latent Space.} 
To ensure optimal space-filling properties and avoid spurious correlations, a joint sampling strategy is executed in a standard uniform latent space. Let $N$ be the total computational budget. A Latin Hypercube Sampling (LHS) design is utilized to generate an $N \times (d_\xi + d_\theta)$ sample matrix $\mathbf{U}$ drawn from an independent standard uniform distribution $\mathcal{U}(0,1)^{d_\xi + d_\theta}$. This matrix is structurally partitioned into two decoupled blocks: the latent aleatory subset $\mathbf{U}_\xi \in \mathbb{R}^{N \times d_\xi}$ and the latent epistemic subset $\mathbf{U}_\Theta \in \mathbb{R}^{N \times d_\theta}$.

\textbf{2. Inverse Isoprobabilistic Mapping.} 
The latent uniform samples are deterministically mapped to the physical domain, formulating the physical aleatory matrix $\mathbf{X}$ and epistemic matrix $\mathbf{\Theta}$. Depending on the complexity of the uncertainty models, the mapping paradigms include:

\textit{Independent Marginals:} For variables with arbitrary independent distributions, the physical column vector $\mathbf{x}_i \in \mathbb{R}^{N \times 1}$ is obtained directly via the element-wise inverse Cumulative Distribution Function (CDF) applied to the corresponding latent column $\mathbf{u}_{\xi, i}$:
\begin{equation}
	\mathbf{x}_i = F_{X_i}^{-1}(\mathbf{u}_{\xi, i}).
\end{equation}

\textit{Hierarchical/Conditioned Variables:} For aleatory variables whose distribution bounds or moments are governed by epistemic parameters (e.g., forming a probability box), the mapping is explicitly conditioned on the mapped epistemic realization $\mathbf{\Theta}$:
\begin{equation}
	\mathbf{x}_i = F_{X_i|\Theta}^{-1}(\mathbf{u}_{\xi, i} \mid \mathbf{\Theta}).
\end{equation}

\textit{Correlated Multivariate Distributions:} To embed complex dependence structures (e.g., via a Gaussian copula with correlation matrix $\mathbf{R}$), an isoprobabilistic transform (such as the Nataf or Rosenblatt transformation) is applied to project the correlated physical variables back into the independent latent uniform space. This ensures the strict mutual independence required for the orthogonality of the PCE basis functions.

\textbf{3. Vectorized System Evaluation.} 
The fully formulated physical sample matrix $\mathbf{X}$ and the epistemic parameters $\mathbf{\Theta}$ are fed into the deterministic computational model $\mathcal{M}$ to yield the system response vector $\mathbf{Y} \in \mathbb{R}^{N \times 1}$:
\begin{equation}
	\mathbf{Y} = \mathcal{M}(\mathbf{X}, \mathbf{\Theta}).
\end{equation}
Because the input matrix is generated via joint LHS and mapped via explicit operations, this stage is entirely free of nested loops, enabling extreme parallelization.

\textbf{4. Bayesian Adaptive Sparse PCE via RJMCMC.} 
Once the system responses are obtained, they are standardized to mitigate scaling issues. A single, global sparse PCE model is then constructed over the joint latent space $\mathbf{U} = [\mathbf{U}_\xi, \mathbf{U}_\Theta]$ using a Reversible Jump Markov Chain Monte Carlo (RJMCMC) algorithm. 

Unlike traditional regression methods that suffer from overfitting and the curse of dimensionality, RJMCMC treats the number of basis functions (model dimension) and their corresponding multi-indices as unknown random variables. At each MCMC iteration, the algorithm proposes trans-dimensional moves—\textit{Birth} (adding a basis function), \textit{Death} (removing one), or \textit{Mutate} (altering the degree or variable partition of an existing basis). The acceptance of these moves is governed by the Metropolis-Hastings ratio, incorporating a Laplace-approximated marginal likelihood and a modified g-prior (or ridge prior) that penalizes overly complex high-order interactions. This dynamic exploration yields a posterior ensemble of sparse polynomial bases and their corresponding coefficients, denoted as $\{ \mathbf{\alpha}^{(s)}, \mathbf{\beta}^{(s)} \}_{s=1}^{N_s}$, where $N_s$ is the number of retained MCMC samples after the burn-in period.

\textbf{5. Analytical Extraction of Conditional Sensitivities.} 
Instead of predicting outputs for new epistemic conditions, the decoupling of uncertainties is transformed into a purely algebraic post-processing step. For any specified epistemic condition $\mathbf{\theta}^*$, the multi-dimensional Legendre basis functions corresponding to the epistemic dimensions are evaluated at the mapped latent point $\mathbf{u}_{\Theta}^*$. By projecting these fixed evaluations onto the remaining aleatory bases across the posterior MCMC samples, closed-form expressions for the conditional variance $\mathbb{V}[\mathbf{Y} \mid \mathbf{\theta}^*]$ and conditional Sobol' indices $S_{i|\mathbf{\theta}^*}$ are derived. The final statistical inference is achieved by averaging these analytical measures across the $N_s$ posterior samples, providing both robust point estimates and their associated Bayesian credible intervals without requiring a single additional call to the physical model. The conditional sensitivity extraction process is summarized in Algorithm~\ref{alg}. 

\begin{algorithm}[H]
	\label{alg}
	\DontPrintSemicolon
	\KwIn{Posterior ensemble $\{ \boldsymbol{\alpha}^{(s)}, \boldsymbol{\beta}^{(s)} \}_{s=1}^{N_s}$, epistemic condition $\boldsymbol{\theta}^*$, aleatory index set $\mathcal{A}_L$}
	\KwOut{Conditional sensitivities ${S}_{u|\boldsymbol{\theta}^*}$ and 95\% credible intervals}
	
	\For{$s = 1$ \textbf{to} $N_s$}{
		1. \textbf{Basis Evaluation:} Compute conditional basis projection for the epistemic part:\\
		$\Psi_{\boldsymbol{\alpha}_K}^{(s)}(\boldsymbol{\theta}^*) \leftarrow \prod_{k \in \mathcal{A}_K} \phi_{\alpha_k}(\theta^*_k)$ \;
		
		2. \textbf{Coefficient Mapping:} Extract conditional aleatory coefficients:\\
		$c_{\boldsymbol{\alpha}_L}^{(s)}(\boldsymbol{\theta}^*) \leftarrow \sum_{\boldsymbol{\alpha}_K} \beta_{(\boldsymbol{\alpha}_K, \boldsymbol{\alpha}_L)}^{(s)} \Psi_{\boldsymbol{\alpha}_K}^{(s)}(\boldsymbol{\theta}^*)$ \;
		
		3. \textbf{Variance Decomposition:} Compute conditional partial variances:\\
		$V_{|\boldsymbol{\theta}^*}^{(s)} = \sum_{\boldsymbol{\alpha}_L \neq \mathbf{0}} (c_{\boldsymbol{\alpha}_L}^{(s)}(\boldsymbol{\theta}^*))^2$ \;
		$V_{u|\boldsymbol{\theta}^*}^{(s)} = \sum_{\boldsymbol{\alpha}_L \in \mathcal{A}_{u,L}} (c_{\boldsymbol{\alpha}_L}^{(s)}(\boldsymbol{\theta}^*))^2$ \;
		
		4. \textbf{Sensitivity Extraction:} Calculate conditional indices:\\
		$S_{u|\boldsymbol{\theta}^*}^{(s)} = V_{u|\boldsymbol{\theta}^*}^{(s)} / V_{|\boldsymbol{\theta}^*}^{(s)}$ \;
	}
	5. \textbf{Posterior Inference:} Aggregate results across the ensemble:\\
	$\bar{S}_{u|\boldsymbol{\theta}^*} = \frac{1}{N_s} \sum_{s=1}^{N_s} S_{u|\boldsymbol{\theta}^*}^{(s)}$ \;
	Compute 2.5\% and 97.5\% quantiles of $\{S_{u|\boldsymbol{\theta}^*}^{(s)}\}_{s=1}^{N_s}$ for credible intervals.
	\caption{Analytical Extraction of Conditional Sobol' Indices via RJPCE}
\end{algorithm}

\section{Cases}
\label{sec4}
\subsection{Analytical benchmark for hybrid uncertainty}
To demonstrate the capability of the proposed Bayesian adaptive conditional PCE framework in decoupling mixed uncertainties, we construct a synthetic parametric stochastic model. This model serves as a rigorous analytical benchmark to verify how the single-surrogate framework accurately extracts continuous, epistemic-resolved sensitivity maps without resorting to nested sampling.

Consider a physical quantity $Y$ driven by two independent aleatory variables $X_1$ and $X_2$, whose distribution bounds are governed by two independent epistemic parameters $a$ and $b$. The system response is explicitly defined by a nonlinear polynomial model:
\begin{equation}
	Y(X_1, X_2) = X_1^2 + X_2^2 + 0.5 X_1 X_2
\end{equation}
where the epistemic parameters are uniformly distributed as $\boldsymbol{\theta} = (a, b) \in [1, 3]^2$. Conditioned on a specific epistemic realization $\boldsymbol{\theta}$, the aleatory variables follow uniform distributions bounded by these parameters: $X_1 \mid a \sim \mathcal{U}[-a, a]$ and $X_2 \mid b \sim \mathcal{U}[-b, b]$. 

Given the independence and zero-mean properties of the conditional uniform distributions, the analytical conditional mean and partial variances admit exact closed-form expressions:
\begin{equation}
	\mathbb{E}[Y \mid \boldsymbol{\theta}] = \frac{a^2 + b^2}{3}
\end{equation}
\begin{equation}
	V_1(\boldsymbol{\theta}) = \frac{4a^4}{45}, \quad
	V_2(\boldsymbol{\theta}) = \frac{4b^4}{45}, \quad
	V_{12}(\boldsymbol{\theta}) = \frac{a^2 b^2}{36}
\end{equation}

The analytical conditional variance is exactly the sum of these partial variances:
\begin{equation}
	\text{Var}[Y \mid \boldsymbol{\theta}] = V_1(\boldsymbol{\theta}) + V_2(\boldsymbol{\theta}) + V_{12}(\boldsymbol{\theta})
\end{equation}

Based on this decomposition, the analytical total-effect conditional Sobol' indices for the aleatory variables $X_1$ and $X_2$ can be formulated as:
\begin{equation}
	S_{T, X_1 \mid \boldsymbol{\theta}} = \frac{V_1(\boldsymbol{\theta}) + V_{12}(\boldsymbol{\theta})}{\text{Var}[Y \mid \boldsymbol{\theta}]}, \quad
	S_{T, X_2 \mid \boldsymbol{\theta}} = \frac{V_2(\boldsymbol{\theta}) + V_{12}(\boldsymbol{\theta})}{\text{Var}[Y \mid \boldsymbol{\theta}]}
\end{equation}
These exact analytical surfaces over the epistemic domain $(a, b)$ provide a rigorous ground truth to assess the accuracy of the proposed RJMCMC-driven conditional PCE extraction.

To evaluate the performance of the proposed analytical conditional framework, a joint Latin Hypercube Sample (LHS) of size $N=5{,}000$ is generated in the 4-dimensional unified latent space $\boldsymbol{U} = (\xi_1, \xi_2, \theta_1, \theta_2) \in [0, 1]^4$. The proposed algorithm constructs a single, unified sparse PCE model over this augmented space using RJMCMC Bayesian learning.

To first quantify the accuracy at the extreme boundaries of the epistemic space, the first-order conditional Sobol' indices are evaluated at the four extreme corners: $\{\min(a), \min(b)\}$, $\{\max(a), \min(b)\}$, $\{\min(a), \max(b)\}$, and $\{\max(a), \max(b)\}$. As depicted in Fig.~\ref{fig:corner_bar}, the stacked bar charts generated from the RJPCE extraction match the analytical truth almost identically. The relative contribution of each aleatory variable shifts dynamically in response to the fixed epistemic bounds: at $\{\max(a), \min(b)\}$, $X_1$ dominates with a sensitivity index of nearly $0.95$, whereas at $\{\min(a), \max(b)\}$, $X_2$ becomes the primary driver. At the symmetric corners, their effects are perfectly balanced, a behavior that the proposed algorithm captures seamlessly through purely algebraic post-processing.

\begin{figure}[htbp]
	\centering
	\includegraphics{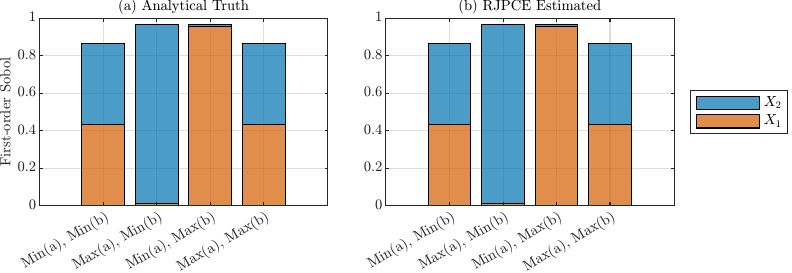}
	\caption{Comparison of first-order conditional Sobol' indices at the four extreme corners of the epistemic parameter space.}
	\label{fig:corner_bar}
\end{figure}

Subsequently, the continuous conditional statistical moments are extracted across a $30 \times 30$ grid within the physical epistemic space. The reconstructed conditional mean and variance, along with their corresponding absolute error fields, are illustrated in Fig.~\ref{fig:cond_moments}. As observed, the conditional mean scales from approximately $0.6$ to $6$, and the variance reaches up to $15$ at the upper bounds of the epistemic space. The RJPCE framework perfectly captures the quadratic dependency of the mean and the quartic dependency of the variance. More importantly, the absolute error fields (right column) confirm the exceptional accuracy of the extraction, with peak errors bounded strictly around $5 \times 10^{-7}$ for the mean and $1.2 \times 10^{-6}$ for the variance.

\begin{figure}[htbp]
	\centering
	\includegraphics{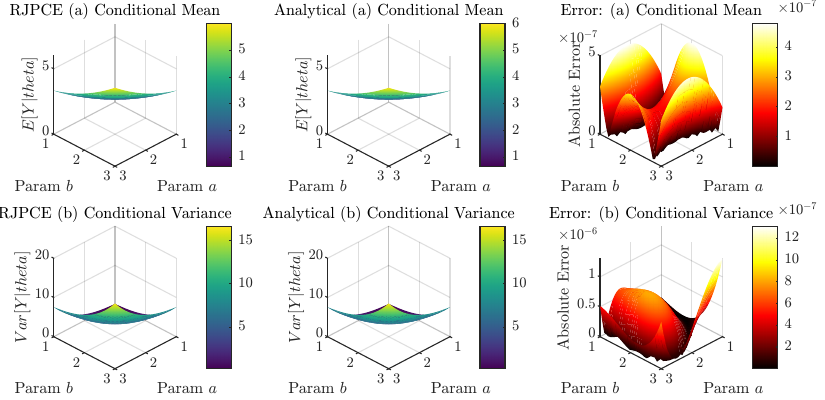}
	\caption{Comparison of the conditional statistical moments over the epistemic domain $(a, b)$. Top row: Conditional Mean $\mathbb{E}[Y \mid \boldsymbol{\theta}]$. Bottom row: Conditional Variance $\text{Var}[Y \mid \boldsymbol{\theta}]$.}
	\label{fig:cond_moments}
\end{figure}

Furthermore, the extracted total-effect conditional Sobol' indices for both $X_1$ and $X_2$ are compared against the exact analytical maps in Fig.~\ref{fig:cond_sobol}. The proposed method produces smooth, spatially coherent sensitivity fields. As physically expected, $S_{T, X_1 \mid \boldsymbol{\theta}}$ approaches $1.0$ when parameter $a$ is maximized and $b$ is minimized, and conversely for $X_2$. The extraction filters out point-wise numerical noise completely, maintaining structural continuity across the entire domain with maximum absolute errors tightly constrained below the $10^{-6}$ magnitude level.

\begin{figure}[htbp]
	\centering
	\includegraphics{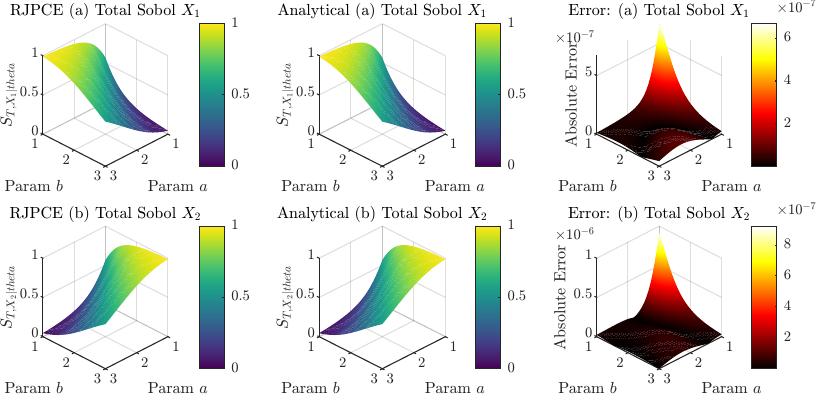}
	\caption{Spatial distributions of the conditional total-effect Sobol' indices over the epistemic domain $(a, b)$. Top row: $S_{T, X_1 \mid \boldsymbol{\theta}}$. Bottom row: $S_{T, X_2 \mid \boldsymbol{\theta}}$.}
	\label{fig:cond_sobol}
\end{figure}

Beyond point estimates, a major advantage of the fully Bayesian RJPCE framework is its ability to naturally quantify the epistemic uncertainty of the surrogate model itself. By evaluating the Sobol' indices across the retained MCMC posterior samples, the 95\% credible intervals (CIs) of the sensitivity measures are analytically extracted. To effectively visualize these bounds—which are otherwise extremely narrow due to the high accuracy of the surrogate—1D sensitivity slices are taken through the center of the epistemic domain (i.e., varying $a$ with fixed $b$, and vice versa). As depicted in Fig.~\ref{fig:cond_sobol_ci}, the widths of the 95\% CIs are scaled by a factor of $1{,}000$ and plotted as shaded bands centered around the RJPCE posterior mean. The results demonstrate an excellent coverage property: the exact analytical solutions (dashed lines) are perfectly enveloped by the credible intervals. This confirms that the proposed framework not only provides high-fidelity point estimates but also rigorously quantifies and bounds its own approximation uncertainty.

\begin{figure}[htbp]
	\centering
	\includegraphics{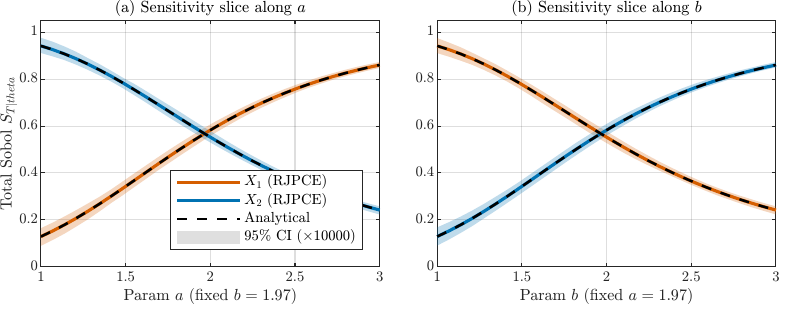}
	\caption{1D sensitivity slices of the total-effect Sobol' indices with 95\% Bayesian credible intervals. The widths of the credible intervals are scaled by a factor of $10{,}000$ for visualization. Left: Varying $a$ with fixed $b=2.0$. Right: Varying $b$ with fixed $a=2.0$.}
	\label{fig:cond_sobol_ci}
\end{figure}

In conclusion, relying on a single set of $5{,}000$ model evaluations, the proposed method successfully constructs a unified representation that not only circumvents the computational bottleneck of double-loop sampling but also ensures physical consistency and high numerical robustness across the varying epistemic conditions.

\subsection{Engineering application: Borehole model with hyperparameter-driven hybrid uncertainty}
To demonstrate the scalability and robustness of the proposed Bayesian adaptive conditional PCE framework in high-dimensional engineering applications, we investigate the classic Borehole groundwater flow model. This model simulates water flow through a borehole connecting two aquifers and is widely utilized as a benchmark for complex uncertainty quantification tasks. The physical schematic of the system is illustrated in Fig.~\ref{fig:bh_schematic}.

\begin{figure}[htbp]
	\centering
	\includegraphics{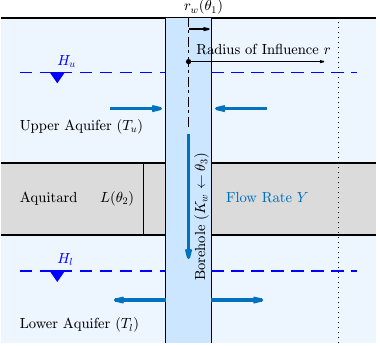}
	\caption{Physical schematic of the Borehole groundwater flow model, illustrating the mapping of hyperparameter-driven epistemic uncertainties to the spatial dimensions and conductivities.}
	\label{fig:bh_schematic}
\end{figure}

We define three epistemic parameters $\boldsymbol{\theta} = (\theta_1, \theta_2, \theta_3)$ governing the spread (half-width) of key aleatory variables. The corresponding uncertainty specification is summarized in Table~\ref{tab:borehole_uncertainty}.

\begin{table}[htbp]
	\centering
	\caption{Hyperparameter-driven hybrid uncertainty specification for the Borehole model (consistent with numerical implementation).}
	\label{tab:borehole_uncertainty}
	\begin{tabular}{lll}
		\hline
		\textbf{Quantity} & \textbf{Distribution} & \textbf{Epistemic control / Parameters} \\
		\hline
		
		Borehole radius $r_w$ 
		& $\mathcal{U}\big[0.10 - \theta_1,\; 0.10 + \theta_1\big]$ 
		& $\theta_1 \in [0.005, 0.05]$ \\
		
		Borehole length $L$ 
		& $\mathcal{U}\big[1400 - \theta_2,\; 1400 + \theta_2\big]$ 
		& $\theta_2 \in [100, 600]$ \\
		
		Hydraulic conductivity $K_w$ 
		& $\mathcal{U}\big[10950 - \theta_3,\; 10950 + \theta_3\big]$ 
		& $\theta_3 \in [500, 4000]$ \\
		
		Upper aquifer head $H_u$ 
		& $\mathcal{U}[950,\;1150]$ 
		& -- \\
		
		Lower aquifer head $H_l$ 
		& $\mathcal{U}[730,\;790]$ 
		&--\\
		
		Transmissivity $T_u$ 
		& deterministic (constant) 
		& $T_u = 89335$ \\
		
		Transmissivity $T_l$ 
		& deterministic (constant) 
		& $T_l = 89.5$ \\
		
		Radius of influence $r$ 
		& deterministic (constant) 
		& $r = 25050$ \\
		
		\hline
	\end{tabular}
\end{table}

The steady-state water flow rate $Y$ is given by:
\begin{equation}
	Y = \frac{2\pi T_u (H_u - H_l)}{\ln(r/r_w) \left[ 1 + \frac{2 L T_u}{\ln(r/r_w) r_w^2 K_w} + \frac{T_u}{T_l} \right]}
\end{equation}

Extracting continuous sensitivity maps for this complex hyperparameter-induced problem using traditional nested Double-Loop Monte Carlo (DLMC) is computationally prohibitive. To achieve a high-resolution evaluation across a $30 \times 30$ epistemic grid with $50{,}000$ MC samples per grid point, DLMC requires $4.5 \times 10^7$ physical model evaluations. In stark contrast, the proposed RJPCE framework constructs a global surrogate over the augmented 8-dimensional latent space using merely $N=5{,}000$ initial LHS samples. This yields an exceptional acceleration ratio of $9{,}000\times$ while maintaining rigorous physical fidelity. 

The reconstructed conditional statistical moments over the epistemic domain $(\theta_1, \theta_2)$ are compared against the high-fidelity DLMC reference in Fig.~\ref{fig:bh_moments}. As the dispersion parameters $\theta_1$ and $\theta_2$ increase, the conditional variance inherently expands. The absolute relative errors are uniformly suppressed to the magnitude of $10^{-3}$, verifying that the global sparse surrogate preserves high accuracy across varying distribution limits without localized retraining.

\begin{figure}[htbp]
	\centering
	\includegraphics{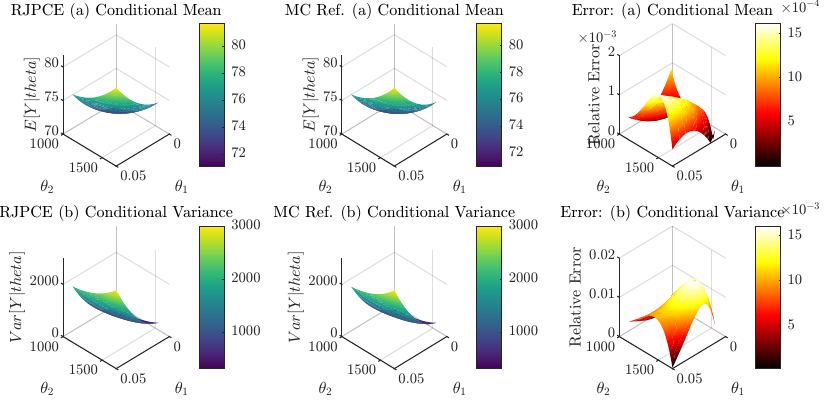}
	\caption{Comparison of the conditional statistical moments for the dynamic Borehole model over the epistemic domain $(\theta_1, \theta_2)$ with $\theta_3$ fixed at its mean value.}
	\label{fig:bh_moments}
\end{figure}

A defining advantage of the fully Bayesian RJPCE framework is its innate capacity to capture complex non-linear variance interactions while quantifying its own surrogate-induced uncertainties. Fig.~\ref{fig:bh_sens_matrix} presents the 1D continuous evolution of the total-effect Sobol' indices for all five aleatory variables. To visualize the algorithm's confidence, the 95\% Bayesian credible intervals (CIs) extracted from the RJMCMC posterior samples are visually enlarged by a factor of $10$ and plotted as shaded bands.

\begin{figure}[htb]
	\centering
	\includegraphics{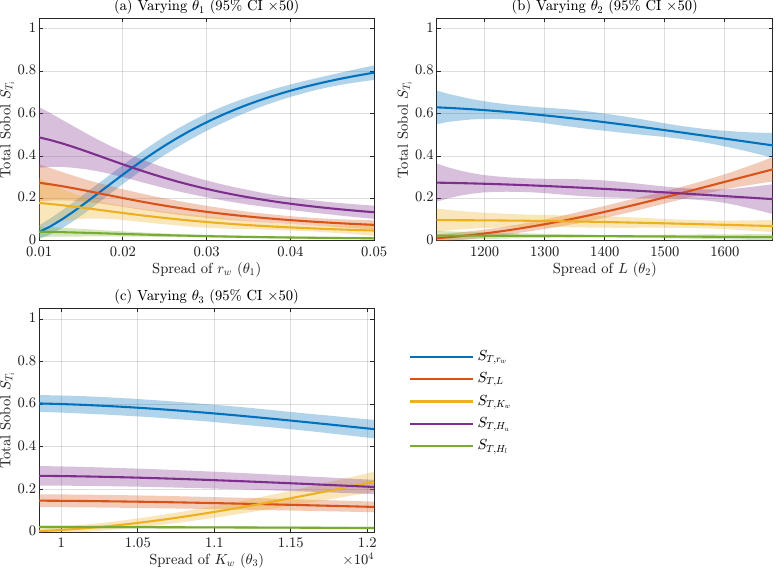}
	\caption{Evolution of conditional total-effect Sobol' indices driven by epistemic hyperparameters. Solid lines indicate the RJPCE posterior means, while shaded regions represent the visually enlarged ($\times 10$) 95\% Bayesian credible intervals. The subplots demonstrate the dynamic variance competition among aleatory variables.}
	\label{fig:bh_sens_matrix}
\end{figure}

The sensitivity evolution in Fig.~\ref{fig:bh_sens_matrix} effectively reveals the underlying ``variance competition'' mechanism of the physical system. Taking Fig.~\ref{fig:bh_sens_matrix}(a) as an example, as $\theta_1$ increases, the variance injected into $r_w$ grows quadratically. Consequently, $S_{T, r_w}$ (blue line) surges non-linearly to dominate the system, which mathematically forces the relative sensitivity contributions of other variables (such as $H_u$ and $L$) to be severely suppressed. Furthermore, the hierarchical separation of $H_u$ (purple line) and $H_l$ (green line) perfectly reflects the asymmetrical bounds assigned to them (width 200 vs. width 60), demonstrating the surrogate's precise resolution in decoupling mixed uncertainties. The tight wrapping of the credible intervals confirms the high robustness of the RJPCE estimation across the entire hybrid parameter space.

\section{Conclusion}
In this paper, a unified Bayesian framework was developed for the analytical extraction of continuous conditional Sobol' indices from a global Polynomial Chaos Expansion (PCE) under hybrid uncertainty. By explicitly decoupling aleatory and epistemic uncertainties via an inverse probability integral transform, the proposed approach preserves both the physical interpretability of the parameters and the mathematical elegance of the tensor-product orthogonal basis.

(1) We formally established that a global PCE surrogate can be analytically factorized. The influence of epistemic variables is entirely encapsulated into parametric coefficient fields, enabling the direct, closed-form derivation of conditional statistical moments and Sobol' indices. This completely bypasses the need for computationally prohibitive double-loop sampling or localized surrogate retraining.
	
(2) By embedding the analytical decomposition within an RJMCMC-driven sparse PCE framework, the method effectively mitigates the curse of dimensionality. Moreover, it simultaneously performs adaptive basis selection and parameter inference, naturally yielding full posterior distributions and rigorous Bayesian credible intervals for the extracted sensitivity measures.
	
(3) Numerical validations, including an exact analytical benchmark and a high-dimensional Borehole model, demonstrated the accuracy and scalability of the proposed approach. The framework not only verified the strict coverage of the Bayesian credible intervals against exact analytical solutions, but also yielded smooth, physically consistent sensitivity maps for complex engineering systems. Crucially, it replaces 9000 repeated surrogate constructions with a single global construction, while preserving the accuracy of the variance competition dynamics among the aleatory variables.

\bibliographystyle{unsrt}  
\bibliography{bib}  

\end{document}